\documentclass[lettersize, journal]{IEEEtran}
\usepackage{amsmath,amsthm,amssymb,amsfonts,bm}
\usepackage{algorithmic}
\usepackage{algorithm}
\usepackage{array}
\usepackage{textcomp}
\usepackage{stfloats}
\usepackage{url}
\usepackage{verbatim}
\usepackage{graphicx}
\usepackage{cite}
\usepackage{xcolor}
\usepackage{utfsym}
\usepackage{pifont}
\usepackage{hyperref}
\usepackage{caption, subcaption}
\usepackage{indentfirst}
\newtheorem{lemma}{Lemma}

\begin{document}

\title{UAV-Enabled Wireless Networks for Integrated Sensing and Learning-Oriented Communication}
\author{Wenhao Zhuang, Xinyu He, Yuyi Mao,~\IEEEmembership{Senior Member,~IEEE}, and Juan Liu,~\IEEEmembership{Member,~IEEE}
	\thanks{
		W. Zhuang, X. He, and Y. Mao are with the Department of Electrical and Electronic Engineering, Hong Kong Polytechnic University, Hong Kong (e-mails: wzhuan@polyu.edu.hk, xinyu.he@connect.polyu.hk, yuyi-eie.mao@polyu.edu.hk). J. Liu is with the School of Electrical Engineering and Computer Science, Ningbo University, Ningbo, China (e-mail: liujuan1@nbu.edu.cn)\emph{(Corresponding author: Yuyi Mao.)} 
		}
}
\maketitle
\begin{abstract}
	Future wireless networks are envisioned to support both sensing and artificial intelligence (AI) services. However, conventional integrated sensing and communication (ISAC) networks may not be suitable due to the ignorance of diverse task-specific data utilities in different AI applications.
	In this letter, a full-duplex unmanned aerial vehicle (UAV)-enabled wireless network providing sensing and edge learning services is investigated. To maximize the learning performance while ensuring sensing quality, a convergence-guaranteed iterative algorithm is developed to jointly determine the uplink time allocation, as well as UAV trajectory and transmit power. Simulation results show that the proposed algorithm significantly outperforms the baselines and demonstrate the critical tradeoff between sensing and learning performance.
\end{abstract}

\section{Introduction}
Integrated sensing and communication (ISAC) has garnered significant attention due to the ever-increasing demand for versatile wireless networks~\cite{fan_isac_2022}. Future sixth-generation (6G) mobile networks are also envisioned to support a variety of edge artificial intelligence (AI) applications~\cite{yuyi_greenai_2024}. To achieve these ambitions, it is critical to investigate multi-functional wireless networks that can simultaneously provide ISAC and AI services. Prior studies on ISAC-empowered wireless networks mostly focus on data-oriented communication~\cite{haocheng_isac_2023, zhenyao_fd_isac_2023}, which targets at conventional quality of service (QoS) metrics such as throughput and latency. However, the transmitted data may carry task-specific utilities for different AI applications. 

In particular, for model training at a mobile edge computing (MEC) server, training data collected by Internet-of-Things (IoT) devices need to be uploaded~\cite{dongzhu_data_2019}. To achieve satisfactory accuracy, radio resources should be properly allocated for data collection in specific model training tasks. This new characteristic has prompted the recent investigations of learning-oriented communication~\cite{shuai_ml_edge_2020, laingkai_edge_2021}. Specifically, learning-centric power allocation algorithms were developed in~\cite{shuai_ml_edge_2020} for edge model training via an empirical function of learning performance with respect to the training data volume. This work was extended in~\cite{laingkai_edge_2021} by further optimizing the uplink time allocation, where an inverse power relationship between generalization error and uplink transmission time was established.

Wireless networks that synergize sensing and edge learning services have also emerged~\cite{mengxuan_fedisac_2024},\cite{peixi_fedisac_2022}. 
A dual-functional wireless network for target sensing and federated learning (FL) was proposed in~\cite{mengxuan_fedisac_2024}, where the receive beamforming and uplink IoT device scheduling are optimized to improve the learning performance while meeting the sensing requirement. In~\cite{peixi_fedisac_2022}, a vertical FL pipeline was presented, which learns to recognize human motion from distributed wireless sensing data. 
Nevertheless, these works do not fit scenarios with low-end IoT devices, where local training is prohibitive due to resource constraints. 
Moreover, sensing targets and IoT devices may be beyond the coverage of base stations. This practical concern warrants the aid of unmanned aerial vehicles (UAVs)~\cite{qingqing_uav_2018}.

In this letter, a novel UAV-enabled wireless network for integrated sensing and learning-oriented communication is investigated. Specifically, a full-duplex UAV is employed to collect data samples from IoT devices for training classification models (CMs) at an MEC server and performs target sensing. Notably, our system is distinct from the UAV-enabled ISAC network~\cite{kaitao_isac_2023} that focuses on data-oriented communication. Compared to the UAV-enabled edge learning system~\cite{jianxin_uav_learn_2023}, our system is more challenging in order to cater both the sensing and learning performance requirements. 
Our main contributions are summarized as follows: First, an optimization problem is formulated for jointly determining the uplink time allocation, as well as UAV trajectory and transmit power. This problem aims to maximize the accuracy of trained CMs subject to the sensing requirement. Second, a block coordinate descent (BCD)-based algorithm with convergence guarantee is developed, which ingeniously applies quadratic transform~\cite{yannan_fp_2024} and convex approximation techniques. Simulation results demonstrate the superiority of the proposed algorithm. It is also found that joint UAV trajectory and transmit power optimization is pivotal for managing the sensing and learning performance.

\textbf{Notations:} Boldface lower-/upper-case letters denote vectors/matrices, and $\mathcal{CN}(\bm{{\mu}},\bm{{\Sigma}})$ represents a complex Gaussian distribution with mean $\bm{\mu}$ and covariance matrix $\bm{{\Sigma}}$. Besides, $(\cdot)^{\text{T}}$, $(\cdot)^{\mathrm{H}}$, $\mathbb{E}\left\{\cdot\right\}$, and $\lVert \cdot \rVert$ denote the transpose, conjugate transpose, statistical expectation, and Euclidean norm, respectively, and $\mathbf{v} \succeq \mathbf{0}$ indicates vector $\mathbf{v}$ is component-wise nonnegative.

\section{System Model}
	\subsection{Network Description}
	We consider a UAV-enabled dual-functional wireless network shown in Fig.~\ref{fig:system_model}, where a full-duplex UAV is employed to collect labeled data samples from $K$ single-antenna IoT devices for training $M$ CMs at an MEC server. The UAV is also responsible for sensing a potential target through the same frequency band. We denote the sets of CMs and IoT devices as $\mathcal{M}\triangleq\left\{1, \cdots, M\right\}$ and $\mathcal{K}\triangleq\left\{1, \cdots, K\right\}$, respectively. 
	The UAV is equipped with half-wavelength spaced uniform linear arrays (ULAs) with $N_t=N_a$ transmit elements and $N_r=N_a$ receive elements, and flies at altitude $H$. It departs from the MEC server at time $0$ and returns at time $T$. After that, the collected data samples are transferred to the MEC server for model training. Following~\cite{shuai_ml_edge_2020}, the IoT devices are partitioned into $M$ non-overlapping groups, denoted as $\left\{\Omega_{1},\cdots, \Omega_{M}\right\}$. Devices in $\Omega_{m}$ possess data samples for training the $m$-th CM. For tractability, the period $T$ is equally divided into $N$ time slots indexed by $n \in \mathcal{N}\triangleq \left\{0, \cdots, N\right\}$. Each time slot has a length of $\delta=T \slash N$. 
	Denote $\mathbf{Q} \triangleq \left\{\mathbf{q}_n\right\}$ as the horizontal coordinates of the UAV. The UAV can be regarded as static at $\mathbf{q}_n$ within time slot $n$ given sufficiently large $N$~\cite{qingqing_uav_2018}.
	The IoT devices and sensing target are on the ground. We denote the maximum velocity of the UAV as $v_{\max}$. The horizontal coordinate of the $k$-th IoT device and sensing target are denoted as $\mathbf{l}_k$ and $\mathbf{t}$, respectively. Moreover, $\mathbf{q}_{0}$ denotes the horizontal coordinate of the MEC server.
	\begin{figure}[t!]
		\centering
		\includegraphics[width=1.0\linewidth]{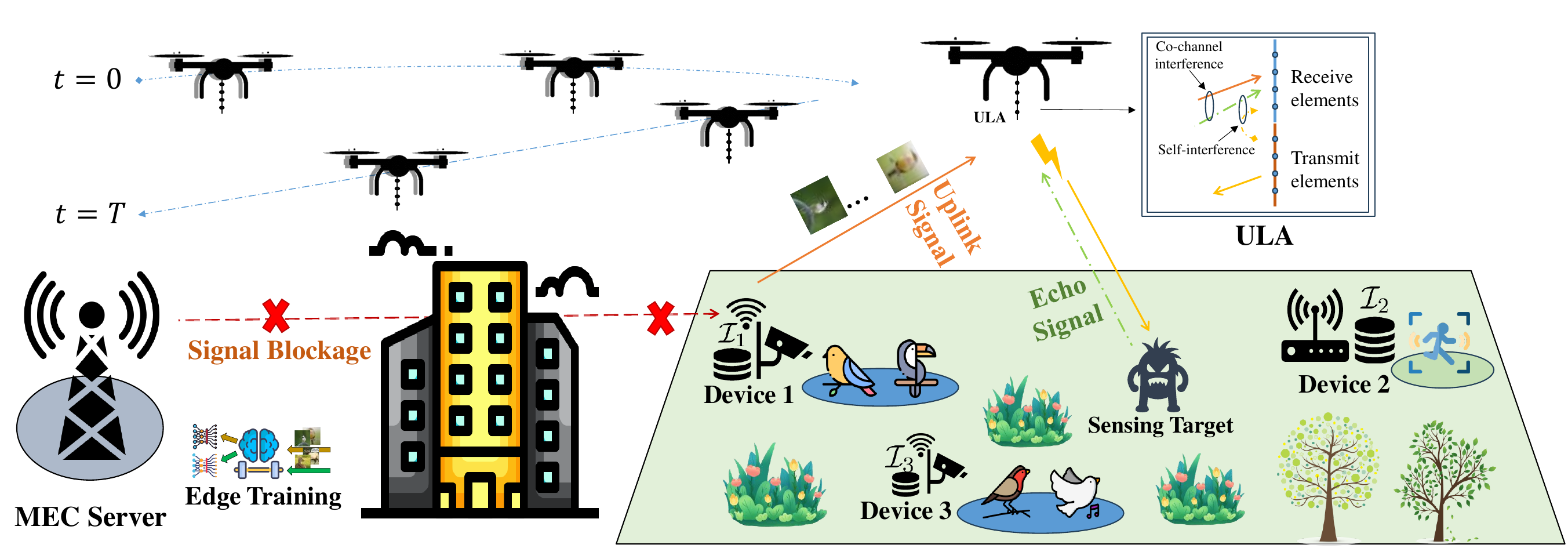}
		\caption{.~~A full-duplex UAV-enabled wireless network for integrated sensing and learning-oriented communication.}
		\label{fig:system_model}
		\vspace{-1.5em}
	\end{figure}
	
	\vspace{-0.5em}
	\subsection{Channel and Signal Models}
	Let $d\left(\mathbf{q}_n, \mathbf{g}\right) \triangleq \sqrt{H^{2}+\|\mathbf{q}_n-\mathbf{g}\|^{2}}$ denote the distance between the UAV and ground location $\mathbf{g}$ at time slot $n$, and $\mathbf{a}(\mathbf{q}_n, \mathbf{g}) \triangleq \big[1, e^{j \pi \frac{H}{d(\mathbf{q}_n, \mathbf{g})}},\cdots, e^{j \pi (N_a-1) \frac{H}{d(\mathbf{q}_n, \mathbf{g})}} \big]^{\mathrm{H}}$ be the steering vector toward $\mathbf{g}$. Assuming line-of-sight (LoS) propagation, the channel vector between the $k$-th IoT device and the UAV, and the round-trip channel matrix for radar sensing are respectively given as $\mathbf{h}_{k}(\mathbf{q}_n) = \sqrt{{\lambda_0}{d^{-2}(\mathbf{q}_n, \mathbf{l}_k)}} \mathbf{a}(\mathbf{q}_n, \mathbf{l}_k)$ and $\mathbf{H}_{\text{r}, n} = \sqrt{{\lambda_0 \xi}{d^{-4}(\mathbf{q}_n, \mathbf{t})}} \mathbf{a}(\mathbf{q}_n, \mathbf{t}) \mathbf{a}^{\mathrm{H}} (\mathbf{q}_n, \mathbf{t})$, where $\lambda_0$ and $\xi$ denote the channel power gain at the reference distance and the radar cross-section, respectively \cite{zhenyao_fd_isac_2023}.

	Time-division multiple access (TDMA) is adopted to coordinate uplink transmissions of IoT devices. Each time slot is partitioned into $K$ segments and $\beta_{k, n}$ is the proportion of time slot $n$ allocated to the $k$-th IoT device.
	Denote $\bm{\beta} \triangleq \left\{\beta_{k, n}\right\}$ and let $s_{k, n} \sim \mathcal{CN}(0, p_k)$ be the uplink signal of the $k$-th IoT device with average power $p_k$. Besides, we denote $\mathbf{p}_{\text{U}} \triangleq \left\{{p}_{\text{U}, n}\right\}$ as the UAV transmit power within all $N$ time slots and the sensing signal transmitted by the UAV as $\mathbf{x}_n = \sqrt{\frac{p_{\text{U}, n}}{N_a}} \mathbf{a}(\mathbf{q}_n, \mathbf{t})$ that maximizes the power at the sensing target~\cite{mimo_radar_2007}. The received signal at the UAV when the $k$-th IoT device is transmitting at time slot $n$ is given as $ \mathbf{y}_{k, n} = \mathbf{h}_{k}(\mathbf{q}_n) s_{k, n} + \mathbf{H}_{\text{r}, n} \mathbf{x}_n + \mathbf{H}_{\text{SI}} \mathbf{x}_n + \mathbf{z}_n$, where $\mathbf{H}_{\text{SI}} \in \mathbb{C}^{N_r \times N_t}$ represents the self-interference (SI) channel matrix and its $(p,q)$-th entry is given by $\sqrt{\alpha_{\text{SI}}} e^{j2\pi \frac{d_{p, q}}{\lambda}}$ with $\alpha_{\text{SI}}$, $\lambda$, and $d_{p, q}$ denoting the SI power coefficient, wavelength, and distance between the $p$-th transmit antenna and $q$-th receive antenna, respectively \cite{xiaoming_si_2018}. 
	We also denote $\mathbf{z}_n \sim \mathcal{CN}(0, \sigma^2 \mathbf{I})$ as the receive noise with $\sigma^{2}$ as the noise power at each antenna.

	\vspace{-0.5em}
	\subsection{Learning-Oriented Communication and Target Sensing}
	The signal-to-interference-plus-noise ratio (SINR) of the received signal from the $k$-th IoT device at time slot $n$ can be expressed as
	\begin{align}
		\tilde{\gamma}_k (\mathbf{q}_n, {p}_{\text{U}, n}) & = \frac{\mathbb{E} \{{\lVert \mathbf{h}_k \left(\mathbf{q}_n\right) s_{k, n} \rVert}^2 \}}{\lVert \left(\mathbf{H}_{\text{r}, n}+ \mathbf{H}_{\text{SI}} \right) \mathbf{x}_n \rVert^2+ \mathbb{E}\{{\lVert \mathbf{z}_n \rVert}^2\} } \nonumber \\
		& = \frac{\lambda_0 p_k N_a d^{-2}(\mathbf{q}_n, \mathbf{l}_k)}{\lVert \left(\mathbf{H}_{\text{r}, n}+ \mathbf{H}_{\text{SI}} \right) \mathbf{x}_n \rVert^2+ N_a \sigma^2},	\label{eq:com_sinr}
	\end{align}
	and the achievable rate is given as $\tilde{R}_k(\mathbf{q}_n, {p}_{\text{U}, n}) = B\log_2 (1 + \tilde{\gamma}_k(\mathbf{q}_n, {p}_{\text{U}, n}))$, where $B$ is the system bandwidth. The number of bits that can be uploaded from the $k$-th IoT device to the UAV over time period $T$ is given as $\tilde{A}_k(\mathbf{Q}, \bm{\beta}, \mathbf{p}_{\text{U}}) \triangleq \delta \sum_{n=1}^{N} \beta_{k, n} \tilde{R}_k(\mathbf{q}_n, p_{\text{U}, n})$, which should comply with the data availability constraint, i.e., $\tilde{A}_{k}(\mathbf{Q}, \bm{\beta}, \mathbf{p}_{\text{U}})\leq \mathcal{I}_{k}D_{m}, m\in\mathcal{M}, k \in \Omega_{m}$. Here, $D_{m}$ denotes the size of a data sample for training the $m$-th CM and $\mathcal{I}_k$ is the number of samples stored at the $k$-th IoT device. The number of data samples collected for training the $m$-th CM is approximated as $\sum_{k \in \Omega_{m}} {\tilde{A}_k(\mathbf{Q}, \bm{\beta}, \mathbf{p}_{\text{U}})}\slash {D_{m}}$. 
	The SINR of the received sensing echo at the UAV while the $k$-th IoT device is transmitting at time slot $n$ is given by $\tilde{\gamma}_k^{\text{rad}}(\mathbf{q}_n, {p}_{\text{U}, n})  = \frac{\lVert \mathbf{H}_{\text{r}, n} \mathbf{x}_n \rVert^2}{\lambda_0 p_k d^{-2}(\mathbf{q}_n, \mathbf{l}_k) + {\lVert \mathbf{H}_{\text{SI}} \mathbf{x}_n \rVert}^2 + N_a \sigma^2}$. Otherwise, $\tilde{\gamma}_k^{\text{rad}}(\mathbf{q}_n, {p}_{\text{U}, n}) = \frac{\lVert \mathbf{H}_{\text{r}, n} \mathbf{x}_n \rVert^2}{{\lVert \mathbf{H}_{\text{SI}} \mathbf{x}_n \rVert}^2 + N_a \sigma^2}$. 
	Next, we maximize the accuracy of the trained CMs while satisfying the sensing requirement.

\section{Problem Formulation}
	Denote $\phi$ as the maximum classification error among CMs. The accuracy of the bottlenecking CMs is maximized by solving the following problem:
	\begin{subequations}
	\begin{align}
		&\textbf{P}_{\text{1}}:  \min_{\bm{\beta},\mathbf{Q}, \mathbf{p}_{\text{U}}, \phi} \ \phi \nonumber \\
		&\ \ \ \ \ \ \ \  \text{s.t. } \ \phi \geq \tilde{\Psi}_{m}(\mathbf{Q}, \bm{\beta}, \mathbf{p}_{\text{U}}), m\in \mathcal{M}, \label{eq:phi} \\ 
		&\ \ \ \ \ \ \ \ \ \ \ \ \ \beta_{k, n}\in[0, 1], k \in \mathcal{K}, n\in\mathcal{N}, \label{eq:time0}\\ 
		&\ \ \ \ \ \ \ \ \ \ \ \ \ \sum\nolimits_{k \in \mathcal{K}} \beta_{k, n}\leq 1, n\in\mathcal{N}, \label{eq:uplinkTimeConstraint} \\
		&\ \ \ \ \ \ \ \ \ \ \ \ \ \tilde{\gamma}^{\text{rad}}_k(\mathbf{q}_n, {p}_{\text{U}, n})\geq \gamma_{\text{th}}^{\text{rad}}, k \in \mathcal{K}, n\in \mathcal{N},  \label{eq:radar_th} \\
		&\ \ \ \ \ \ \ \ \ \ \ \ \ \tilde{A}_{k}(\mathbf{Q}, \bm{\beta}, \mathbf{p}_{\text{U}})\leq \mathcal{I}_{k}D_{m}, m\in\mathcal{M}, k \in \Omega_{m}, \label{eq:data_availability} \\
		&\ \ \ \ \ \ \ \ \ \ \ \ \ p_{\text{U}, n} \leq p_{\text{UAV}}, n\in \mathcal{N}, \label{eq:power} \\
		&\ \ \ \ \ \ \ \ \ \ \ \ \ \mathbf{q}_0 = \mathbf{q}_n, \label{eq:traj_end} \\
		&\ \ \ \ \ \ \ \ \ \ \ \ \ \lVert \mathbf{q}_n-\mathbf{q}_{n-1}\rVert^2 \leq \left(v_{\max} \delta\right)^{2}, n\in\mathcal{N}\setminus \left\{0\right\}, \label{eq:mobility}
	\end{align}
	\end{subequations}
	where $\tilde{\Psi}_{m}(\mathbf{Q}, \bm{\beta}, \mathbf{p}_{\text{U}}) = a_{m}(\sum_{n=1}^{N} \sum_{k \in \Omega_{m}} \beta_{k,n} \tilde{R}_k(\mathbf{q}_n, p_{\text{U}, n})$ $ \cdot \delta \slash D_{m}  + A_m)^{-b_{m}}$ is a surrogate of classification error with $a_{m}$, $b_{m}$ being non-negative parameters specific to training CM $m$, and $A_m\geq 0$ denoting the number of historical data samples available at the MEC server prior to the UAV's flight~\cite{shuai_ml_edge_2020}. 
	The values of $\{a_m\}$ and $\{b_m\}$ can be obtained via curve fitting by training preliminary models on the historical data. In $\textbf{P}_{\text{1}}$, \eqref{eq:time0}, \eqref{eq:uplinkTimeConstraint} are the uplink time allocation constraints, and \eqref{eq:radar_th} enforces the sensing SINR to be above the threshold $\gamma^{\text{rad}}_{\text{th}}$. Constraint \eqref{eq:power} limits the transmit power of the UAV to $p_{\text{UAV}}$ and \eqref{eq:traj_end} specifies the initial/final position of the UAV. Moreover, \eqref{eq:mobility} is the UAV displacement constraint. 

	Since $\mathbf{Q}$, $\mathbf{p}_{\text{U}}$, and $\mathbf{H}_{\text{SI}}$ are entangled in $\tilde{\gamma}_k(\mathbf{q}_n, {p}_{\text{U}, n})$, we apply $\lVert \left(\mathbf{H}_{\text{r}, n} + \mathbf{H}_{\text{SI}} \right) \mathbf{x}_n \rVert \leq  \lVert\mathbf{H}_{\text{r}, n} \mathbf{x}_n \rVert + \lVert \mathbf{H}_{\text{SI}} \mathbf{x}_n \rVert \leq \sqrt{\lambda_0 \xi N^2_a d^{-4}(\mathbf{q}_n, \mathbf{l}_k) p_{\text{U}, n}} + \sqrt{\alpha_{\text{SI}} N^2_a  p_{\text{U}, n}}$ 
	to derive a lower bound of $\tilde{\gamma}_k(\mathbf{q}_n, {p}_{\text{U}, n})$ as
	\begin{align}
		{\gamma}_k(\mathbf{q}_n, {p}_{\text{U}, n})  & = \frac{\lambda_k  d^{-2}(\mathbf{q}_n, \mathbf{l}_k)}{{(\sqrt{\lambda_t}d^{-2}(\mathbf{q}_n, \mathbf{t}) + \sqrt{\lambda_{\text{SI}}} )}^2 p_{\text{U}, n} + \sigma^2 },
	\end{align}
	where $\lambda_{\text{SI}} \triangleq \alpha_{\text{SI}} N_a$, $\lambda_t \triangleq \lambda_0 \xi N_a$, and $\lambda_k \triangleq \lambda_0 p_k$.
	Thus, the data rate is lower bounded by $R_k(\mathbf{q}_n, p_{\text{U}, n}) = B\log_2(1 + \gamma_k(\mathbf{q}_n, p_{\text{U}, n}))$. Likewise, $\tilde{\gamma}_k^{\text{rad}}(\mathbf{q}_n, p_{\text{U}, n})$ is lower bounded by
	\begin{equation}
		\gamma_k^{\text{rad}}(\mathbf{q}_n, {p}_{\text{U}, n}, \beta_{k, n}) = 
		\begin{cases}
			\frac{\lambda_t p_{\text{U}, n} d^{-4}(\mathbf{q}_n, \mathbf{t})}{\lambda_k d^{-2}(\mathbf{q}_n, \mathbf{l}_k) + \lambda_{\text{SI}} p_{\text{U}, n} + \sigma^2},  &{\beta_{k, n} > 0} \\
			\frac{\lambda_t p_{\text{U}, n} d^{-4}(\mathbf{q}_n, \mathbf{t})}{\lambda_{\text{SI}}p_{\text{U}, n} + \sigma^2}, &{\beta_{k, n} = 0}
		\end{cases}.
	\end{equation}
	We replace $\tilde{\Psi}_{m}(\mathbf{Q}, \bm{\beta}, \mathbf{p}_{\text{U}})$ in~\eqref{eq:phi} and $\tilde{A}_{k}(\mathbf{Q}, \bm{\beta}, \mathbf{p}_{\text{U}})$ in~\eqref{eq:data_availability} respectively by ${\Psi}_{m}(\mathbf{Q}, \bm{\beta}, \mathbf{p}_{\text{U}})$ and ${A}_{k}(\mathbf{Q}, \bm{\beta}, \mathbf{p}_{\text{U}})$ (which are obtained by substituting $\tilde{R}_k(\mathbf{q}_n, p_{\text{U}, n})$ with $R_k(\mathbf{q}_n, p_{\text{U}, n})$), and $\tilde{\gamma}^{\text{rad}}_k(\mathbf{q}_n, p_{\text{U}, n})$ in~\eqref{eq:radar_th} by $\gamma_k^{\text{rad}}(\mathbf{q}_n, p_{\text{U}, n}, \beta_{k, n})$ to formulate $\textbf{P}_2$.
	Since ${\gamma}_k(\mathbf{q}_n, {p}_{\text{U}, n})$ and $\gamma_k^{\text{rad}}(\mathbf{q}_n, p_{\text{U}, n}, \beta_{k, n})$ avoid $\mathbf{Q}$ being in the exponent of complex exponential functions, $\textbf{P}_{2}$ is easier to handle. Although the optimal solution of $\textbf{P}_2$ may not meet~\eqref{eq:data_availability}, it can be easily rectified without degrading the objective value.
	
	However, $\textbf{P}_{\text{2}}$ is still difficult to solve. First, $R_k(\mathbf{q}_n, {p}_{\text{U}, n})$ and $\gamma_k^{\text{rad}}(\mathbf{q}_n, {p}_{\text{U}, n}, \beta_{k, n})$ are not jointly convex in $\bm{\beta},\mathbf{Q}$, and $\mathbf{p}_{\text{U}}$. Hence, $\textbf{P}_{2}$ is non-convex. Second, since $\gamma_k^{\text{rad}}(\mathbf{q}_n, {p}_{\text{U}, n}, \beta_{k, n})$ depends on $\mathbf{Q}$, $\bm{\beta}$, and $\textbf{p}_{\text{U}}$, achieving the sensing SINR is challenging without a holistic design of all these variables. Next, we develop an algorithm that guarantees a local optimal solution for $\textbf{P}_{2}$. Then, a feasible solution for $\textbf{P}_{1}$ is constructed.
	
\section{Proposed Solution}
In this section, we develop a BCD-based algorithm to alternatively optimize the uplink time allocation, UAV trajectory, and UAV transmit power, in which, quadratic transform and convex approximation techniques are applied.

\vspace{-0.8em}
\subsection{Uplink Time Allocation}
Denote $\mathbf{Q}^{(i-1)}\triangleq \{\mathbf{q}_n^{(i-1)}\}$ and $\mathbf{p}_{\text{U}}^{(i-1)} \triangleq \{p^{(i-1)}_{\text{U}, n}\}$ as the UAV trajectory and transmit power after the $(i-1)$-th BCD iteration, respectively. In the $i$-th iteration, the uplink time allocation is optimized by solving the following problem:
\begin{subequations}
\begin{align}
	&\textbf{P}^{(i)}_{\text{3}}:  \min_{\bm{\beta}, \phi}\ \phi \nonumber\\
	&\ \ \ \ \ \ \ \text{s.t. } \eqref{eq:time0},\eqref{eq:uplinkTimeConstraint}, \nonumber \\
	&\ \ \ \ \ \ \ \ \ \ \ \ \phi \geq {\Psi}_{m}(\mathbf{Q}, \bm{\beta}, \mathbf{p}_{\text{U}}), m\in \mathcal{M}, \label{eq:phi2} \\
	&\ \ \ \ \ \ \ \ \ \ \ \ {\gamma}^{\text{rad}}_k(\mathbf{q}_n, {p}_{\text{U}, n}, \beta_{k, n})\geq \gamma_{\text{th}}^{\text{rad}}, k \in \mathcal{K}, n\in \mathcal{N},  \label{eq:radar_th2} \\
	&\ \ \ \ \ \ \ \ \ \ \ \ {A}_{k}(\mathbf{Q}, \bm{\beta}, \mathbf{p}_{\text{U}})\leq \mathcal{I}_{k}D_{m}, m\in\mathcal{M}, k \in \Omega_{m}, \label{eq:data_availability2}
\end{align}
\end{subequations}
with $\mathbf{Q} = \mathbf{Q}^{(i-1)}$ and $\mathbf{p}_{\text{U}} = \mathbf{p}^{(i-1)}_{\text{U}}$.
Since \eqref{eq:time0}, \eqref{eq:uplinkTimeConstraint}, \eqref{eq:radar_th2}, \eqref{eq:data_availability2} given $\mathbf{Q}^{(i-1)}$, $\mathbf{p}_{\text{U}}^{(i-1)}$ are linear and $\Psi_m(\mathbf{Q}^{(i-1)}, \bm{\beta}, \mathbf{p}^{(i-1)}_{\text{U}})$ is convex, $\textbf{P}^{(i)}_{\text{3}}$ is a convex problem that can be optimally solved via the interior-point method (IPM).

\subsection{UAV Trajectory Optimization}
With $\bm{\beta} = \bm{\beta}^{(i)}$ and $\mathbf{p}_{\text{U}} = \mathbf{p}^{(i-1)}_{\text{U}}$, the UAV trajectory can be optimized by solving the following non-convex problem:
\begin{subequations}
	\begin{align}
		&\textbf{P}_{\text{4}}^{(i)}:  \min_{\mathbf{Q}, \phi} \phi \ \ \ \ \ \text{s.t. } \eqref{eq:traj_end}, \eqref{eq:mobility}, \eqref{eq:phi2}, \eqref{eq:radar_th2},  \eqref{eq:data_availability2}, \nonumber
	\end{align}
\end{subequations}
for which, lower bounds of $\gamma_k^{\text{rad}}(\mathbf{q}_n, p^{(i-1)}_{\text{U}, n}, \beta^{(i)}_{k, n})$ and $R_k(\mathbf{q}_n, p^{(i-1)}_{\text{U}, n})$ are derived by exploiting the quadratic transform \cite{yannan_fp_2024} in the following lemma.
\begin{lemma}
	\label{lemma:quad_trans_lb}
	Given $\bm{\alpha} \succeq \bm{0}$, $f_n(\bm{x})$, $g_n(\bm{x})>0, \forall \bm{x}$ and $n\in\mathcal{N}$, 
	\begin{align}
		\frac{f_n(\bm{x})}{g_n(\bm{x})} &\geq \left(2\alpha_n \sqrt{f_n(\bm{x})} - \alpha_n^2 g_n(\bm{x}) \right) \triangleq \Upsilon(\bm{\alpha}, \bm{x}) \label{eq:quad_trans_lb},
	\end{align}
	where the equality is achieved at $\alpha_n= {\sqrt{f_n(\bm{x})}}\slash {g_n(\bm{x})}, \forall n$.
\end{lemma}
\begin{proof}
	Since $\Upsilon(\bm{\alpha}, \bm{x})$ is concave in $\bm{\alpha}$, the proof is obtained by solving $\nabla_{\bm{\alpha}} \Upsilon \left(\bm{\alpha}, \bm{x}\right)=\mathbf{0}$.
\end{proof}
Define $l^{(i)}_{n}(\mathbf{q}_n) \triangleq \sqrt{\lambda_t p^{(i-1)}_{\text{U}, n}} d^{-2}\left(\mathbf{q}_n, \mathbf{t}\right)$ and $z^{(i)}_{k, n}(\mathbf{q}_n) \triangleq \lambda_k d^{-2}\left(\mathbf{q}_n, \mathbf{l}_k \right) + \lambda_{\text{SI}} p^{(i-1)}_{\text{U}, n} + \sigma^2$. When $\beta^{(i)}_{k, n} > 0$, a lower bound of $\gamma_k^{\text{rad}}(\mathbf{q}_n, p^{(i-1)}_{\text{U}, n}, \beta^{(i)}_{k, n})$ is derived using Lemma \ref{lemma:quad_trans_lb} as
\begin{align}
	\Phi^{(i)}_{k, n}(\mathbf{q}_n) \triangleq 2 \varphi^{(i)}_{k, n} {l^{(i)}_{n}(\mathbf{q}_n)} - {\varphi^{(i)^2}_{k, n}} z^{(i)}_{k, n}(\mathbf{q}_n),
\end{align}
where $\varphi^{(i)}_{k, n} \! = \! {{l^{(i)}_{n}(\mathbf{q}^{(i-1)}_n)}} \! \slash \! {z^{(i)}_{k, n}(\mathbf{q}^{(i-1)}_n)}$. Although $\Phi^{(i)}_{k, n}(\mathbf{q}_n)$ is non-convex, it can be handled via convex approximation.

Since $d^{-2}(\mathbf{q}_n, \mathbf{g})$ and $d^2(\mathbf{q}_n, \mathbf{g})$ are respectively convex in $d^2(\mathbf{q}_n, \mathbf{g})$ and $d(\mathbf{q}_n, \mathbf{g})$, we derive their lower bounds via the first-order Taylor approximation at $\mathbf{q}_{n} = \mathbf{q}_{n}^{(i-1)}$ as follows:
\begin{align}
	\underline{d}_a^{(i)}(\mathbf{q}_n, \mathbf{g}) &\triangleq \frac{2}{d^{2}(\mathbf{q}^{(i-1)}_{n}, \mathbf{g})}-\frac{d^2(\mathbf{q}_n, \mathbf{g})}{d^{4}(\mathbf{q}^{(i-1)}_{n}, \mathbf{g})}, \\
	\underline{d}_b^{(i)}(\mathbf{q}_n, \mathbf{g}) &\triangleq 2(\mathbf{q}_{n}^{(i-1)}-\mathbf{g}) (\mathbf{q}_n-\mathbf{q}_n^{(i-1)})^{\mathrm{T}} + d^2(\mathbf{q}^{(i-1)}_{n}, \mathbf{g}).
\end{align}
A concave lower bound of $l^{(i)}_{n}(\mathbf{q}_n)$ is thus given as $\underline{l}_{n}^{(i)}(\mathbf{q}_n) \triangleq \sqrt{\lambda_t p^{(i-1)}_{\text{U}, n}} \underline{d}_{a}^{(i)}(\mathbf{q}_n, \mathbf{t})$.
We define auxiliary variables $\mathbf{E} \triangleq \left\{e_k [n] \right\}$ with $ 0 < e_k[n] \leq d^2(\mathbf{q}_n, \mathbf{l}_k)$ to upper bound $z^{(i)}_{k, n}(\mathbf{q}_n)$ as $\overline{z}^{(i)}_{k, n}(\mathbf{q}_n) \triangleq \lambda_k e^{-1}_k [n] +  \lambda_{\text{SI}} p^{(i-1)}_{\text{U}, n} + \sigma^2$.
When $\beta_{k, n}^{(i)} = 0$, $\overline{z}^{(i)}_{k, n}(\mathbf{q}_n) \triangleq \lambda_{\text{SI}} p^{(i-1)}_{\text{U}, n} + \sigma^2$. Thus, \eqref{eq:radar_th2} is tightened as
	\begin{align}
		\underline{\gamma}^{(i)}_{\text{s}, k}(\mathbf{q}_n) &\geq \gamma_{\text{th}}^{\text{rad}}, k \in \mathcal{K}, n\in \mathcal{N},  \label{eq:cvx_rad0} \\
		0< e_k[n] &\leq \underline{d}^{(i)}_b(\mathbf{q}_n, \mathbf{l}_k), k \in \mathcal{K}, n\in \mathcal{N}, \label{eq:cvx_rad1}
	\end{align}
where $\underline{\gamma}^{(i)}_{\text{s}, k}(\mathbf{q}_n) \triangleq 2 \varphi^{(i)}_{k, n} \underline{l}_{n}^{(i)}(\mathbf{q}_n) - {\varphi^{{(i)}^2}_{k, n}} \overline{z}^{(i)}_{k, n}(\mathbf{q}_n)$ is a tight concave lower bound of $\gamma_k^{\text{rad}}(\mathbf{q}_n, p^{(i-1)}_{\text{U}, n}, \beta^{(i-1)}_{k, n})$.

Next, we define $r_{k}(\mathbf{q}_n) \triangleq \lambda_k d^{-2}(\mathbf{q}_n, \mathbf{l}_{k})$ and $\psi^{(i)}_{n}(\mathbf{q}_n) \triangleq (\sqrt{\lambda_t}d^{-2}(\mathbf{q}_n, \mathbf{t}) + \sqrt{\lambda_{\text{SI}}})^2 p^{(i-1)}_{\text{U}, n} + \sigma^2$. 
Thus, $\gamma_k(\mathbf{q}_n, p^{(i-1)}_{\text{U}, n})$ is lower bounded by $\underline{\gamma}^{(i)}_{c, k}(\mathbf{q}_n) \triangleq 2 \rho^{(i)}_{k,  n} \sqrt{r_{k}(\mathbf{q}_n)} - {\rho^{(i)^2}_{k, n}} \psi^{(i)}_{n}(\mathbf{q}_n)$, where $\rho^{(i)}_{k, n} \triangleq {\sqrt{r_{k}(\mathbf{q}^{(i-1)}_n)}} \slash {\psi^{(i)}_{n}(\mathbf{q}^{(i-1)}_n)}$. 
To convexify $\underline{\gamma}^{(i)}_{c, k}(\mathbf{q}_n)$, we derive a concave lower bound of $r_{k}(\mathbf{q}_n)$ as $\underline{r}_{k}^{(i)}(\mathbf{q}_n) \triangleq \lambda_k \underline{d}_a^{(i)}(\mathbf{q}_n, \mathbf{l}_k)$. 
Auxiliary variables $\mathbf{U} \triangleq \left\{u[n]\right\}$ with $0 < u[n] \leq d^2(\mathbf{q}_n, \mathbf{t})$ are introduced and $\psi^{(i)}_{n}(\mathbf{q}_n)$ is upper bounded by $\overline{\psi}^{(i)}_{n}(\mathbf{q}_n) \triangleq {(\sqrt{\lambda_t}u^{-1}[n] + \sqrt{\lambda_{\text{SI}}} )}^2 p^{(i-1)}_{\text{U}, n} + \sigma^2$. 
Thus, $\underline{\gamma}^{(i)}_{k}(\mathbf{q}_n) \triangleq \rho^{(i)}_{k, n}  \sqrt{\underline{r}_{k}(\mathbf{q}_n)} - {\rho^{{(i)}^2}_{k, n}} \overline{\psi}^{(i)}_{n}(\mathbf{q}_n)$ is obtained as a lower bound of $\underline{\gamma}^{(i)}_{k}(\mathbf{q}_n)$ and \eqref{eq:phi2} is tightened as follows: 
\begin{align}
	0 < u[n] &\leq \underline{d}^{(i)}_{b}(\mathbf{q}_n, \mathbf{t}), n\in \mathcal{N}, \label{eq:cvx_un} \\
	\underline{\gamma}^{(i)}_{k}(\mathbf{q}_n) &\geq 0, k \in \mathcal{K}, n\in \mathcal{N}, \\
	{\Psi}^{(i)}_{\mathbf{Q}, m}(\mathbf{Q}) &\leq \phi , m \in \mathcal{M}, \label{eq:cvx_phi_traj}
\end{align}
where ${\Psi}^{(i)}_{\mathbf{Q}, m}(\mathbf{Q})$ is obtained by substituting $\tilde{R}_k(\mathbf{q}_n, p_{\text{U}, n})$ in $\tilde{\Psi}_{m}(\mathbf{Q}, \bm{\beta}, \mathbf{p}_{\text{U}})$ with $\underline{R}^{(i)}_k(\mathbf{q}_n) \triangleq B\log_2 (1 + \underline{\gamma}^{(i)}_{k}(\mathbf{q}_n))$, i.e., a tight lower bound of $R_k(\mathbf{q}_n, p^{(i-1)}_{\text{U}, n})$.

To tackle the non-convex constraint \eqref{eq:data_availability2}, we apply the inverse quadratic transform \cite{yannan_fp_2024} in the following lemma.
\begin{lemma}
	\label{lemma:quad_trans_ub}
	Given $\bm{1} \succeq \bm{\theta} \succeq \bm{0}$, $\bm{\rho} \succeq \bm{0}$, $f_n(\bm{x})$, $g_n(\bm{x})>0, \forall \bm{x}$, $2 \rho_n \sqrt{g_n(\bm{x})} - \rho_n^2 f_n(\bm{x}) \geq 0$ and $n\in\mathcal{N}$, 
	\begin{align}
		\ln \left(1+ \frac{f_n(\bm{x})}{g_n(\bm{x})}\right) &\leq {\left(1-\theta_{n}\right) \frac{f_n(\bm{x})}{g_n(\bm{x})} - h(\theta_n)} \nonumber \\
		& \leq \frac{1-\theta_{n}}{2 \rho_n \sqrt{g_n(\bm{x})} - \rho_n^2 f_n(\bm{x})} - h(\theta_n), \label{eq:ldt}
	\end{align}
	where $h(\theta_n) \triangleq \theta_{n}+ \ln (1-\theta_{n})$. The first and second equalities are respectively achieved at $\theta_n={f_n(\bm{x})} \slash {(f_n(\bm{x}) + g_n(\bm{x}))}$ and $\rho_n={\sqrt{g_n(\bm{x})}}\slash {f_n(\bm{x})}, \forall n$.
\end{lemma}
\begin{proof}
	The proof is similar as that of Lemma \ref{lemma:quad_trans_lb}, which is omitted for brevity.
\end{proof}
\begin{figure}[t!]
	\centering
	\includegraphics[width=0.95\linewidth]{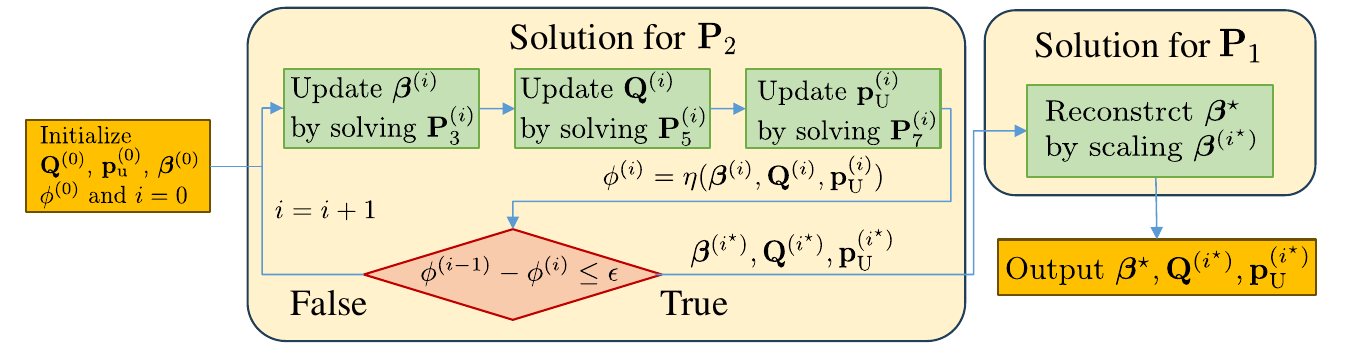}
	\caption{.~~Block diagram of the proposed algorithm for $\textbf{P}_1$.}
	\label{fig:block_diagram}
	\vspace{-1.5em}
\end{figure}

By applying Lemma \ref{lemma:quad_trans_ub}, $R_k(\mathbf{q}_n, p^{(i-1)}_{\text{U}, n})$ is upper bounded by ${R}^{(i)}_{c, k}(\mathbf{q}_n) \! \triangleq \! {B}{\log_2 e} \! \cdot \! \left(\frac{1 - \mu^{(i)}_{k, n}}{2 \nu^{(i)}_{k, n} \sqrt{\psi^{(i)}_{n}(\mathbf{q}_n)} - {\nu^{(i)^2}_{k, n}} r_{k}(\mathbf{q}_n)}\! - \!  h(\mu^{(i)}_{k, n})\right)$, where $\mu^{(i)}_{k, n} \triangleq \frac{r_{k}(\mathbf{q}^{(i-1)}_n)} {r_{k}(\mathbf{q}^{(i-1)}_n) + \psi^{(i)}_{n}(\mathbf{q}^{(i-1)}_n)}$ and $\nu^{(i)}_{k, n} \triangleq \frac{\sqrt{{\psi^{(i)}_{n}(\mathbf{q}^{(i-1)}_n)}}}{r_{k}(\mathbf{q}^{(i-1)}_n)}$. 
Then, $\omega^{(i)}_{k}(\mathbf{q}_n) \triangleq {2 \nu^{(i)}_{k, n} \sqrt{\underline{\psi}^{(i)}_{n}(\mathbf{q}_n)} - {\nu^{(i)^2}_{k, n}} \overline{r}_{k}(\mathbf{q}_n)}$ is derived with $\overline{r}_{k}(\mathbf{q}_n) \triangleq \lambda_k e^{-1}_k[n]$ and $\underline{\psi}^{{(i)}}_{n}(\mathbf{q}_n) \triangleq (\lambda_t \underline{d}_{t}^{(i)}(\mathbf{q}_n) + 2\sqrt{\lambda_t \lambda_{\text{SI}}} \underline{d}_{b}^{(i)}(\mathbf{q}_n, \mathbf{t}) + \lambda_{\text{SI}}) p^{(i-1)}_{\text{U}, n} + \sigma^2$, 
where $\underline{d}_{t}^{(i)}(\mathbf{q}_n) \triangleq 3d^{-4}(\mathbf{q}^{(i-1)}_n, \mathbf{t}) - 2d^2(\mathbf{q}_n, \mathbf{t}) d^{-6}(\mathbf{q}^{(i-1)}_n, \mathbf{t}) $ is a tight concave lower bound of $d^{-4}(\mathbf{q}_n, \mathbf{t})$ obtained via first-order Taylor approximation at $d^{2} \left(\mathbf{q}_{n},\mathbf{t}\right)=d^2(\mathbf{q}^{(i-1)}_n, \mathbf{t})$. Thus, ${R}^{(i)}_{c, k}(\mathbf{q}_n) \leq \overline{R}^{(i)}_{k}(\mathbf{q}_n) \! \triangleq \!  {B}{\log_2 e} \cdot \left(\frac{1 - \mu^{(i)}_{k, n}}{\omega^{(i)}_{k}(\mathbf{q}_n)} -  h(\mu^{(i)}_{k, n})\right)$ and \eqref{eq:data_availability2} is tightened as
\begin{align}
	\omega_{k}^{(i)}(\mathbf{q}_n) &\geq 0, n\in \mathcal{N}, k \in \mathcal{K}, \label{eq:cvx_da_0} \\
	\overline{A}_k(\mathbf{Q}, \bm{\beta}^{(i)}, \mathbf{p}_{\text{U}}^{(i-1)}) &\leq \mathcal{I}_k D_m, m\in\mathcal{M}, k\in \Omega_m, \label{eq:cvx_da_1}
\end{align}
where $\overline{A}_k(\mathbf{Q}, \bm{\beta}^{(i)}, \mathbf{p}_{\text{U}}^{(i-1)}) \triangleq \delta \sum_{n\in\mathcal{N}} \beta_{k,n} \overline{R}^{(i)}_k(\mathbf{q}_n)$. Therefore, $\textbf{P}_3^{(i)}$ is tightened as the following convex problem:
\begin{subequations}
	\begin{align}
		&\textbf{P}_{\text{5}}^{(i)}:\min_{\mathbf{Q}^{(i)}, \phi, \mathbf{E}, \mathbf{U}} \phi \ \ \ \text{s.t. }\eqref{eq:traj_end}, \eqref{eq:mobility}, \eqref{eq:cvx_rad0}\mbox{-}\eqref{eq:cvx_phi_traj}, \eqref{eq:cvx_da_0}, \eqref{eq:cvx_da_1}, \nonumber
	\end{align}
\end{subequations}
which can be optimally solved via IPM. 

\subsection{UAV Transmit Power Optimization}
Given $\bm{\beta}^{(i)}$, $\mathbf{Q}^{(i)}$, the UAV transmit power is optimized by solving the following non-convex problem:
\begin{subequations}
	\begin{align}
		&\textbf{P}_{\text{6}}^{(i)}:  \min_{\mathbf{p}_{\text{U}}^{(i)}, \phi} \phi \ \ \ \ \ \text{s.t. } \eqref{eq:power}, \eqref{eq:phi2}, \eqref{eq:radar_th2}, \eqref{eq:data_availability2}. \nonumber
	\end{align}
\end{subequations}
Define $\zeta^{(i)}_{n} \triangleq(\sqrt{\lambda_t}d^{-2}(\mathbf{q}^{(i)}_n, \mathbf{t}) + \sqrt{\lambda_{\text{SI}}})^2$ and $\kappa^{(i)}_{k, n} \triangleq \lambda_k d^{-2}(\mathbf{q}_{n}^{(i)}, \mathbf{l}_k )$. Since $R_k(\mathbf{q}^{(i)}_n, p_{\text{U}, n})$ is convex in $p_{\text{U}, n}$, we derive its lower bound via first-order Taylor approximation as
\begin{align}
	&{R}_{\mathbf{p}, k}^{(i)}(p_{\text{U}, n}) \triangleq R_k(\mathbf{q}^{(i)}_n, p^{(i-1)}_{\text{U}, n}) \! + \! \left[ (\kappa_{k,n}^{(i)} + \zeta^{(i)}_{n} p_{\text{U}, n}^{(i-1)} + \sigma^{2})^{-1} \right. \nonumber \\
	&\left. - (\kappa_{k, n}^{(i)} p_{\text{U}, n}^{(i-1)} \! + \! \sigma^{2})^{-1} \right]{\zeta_{n}^{(i)}} (p_{\text{U}, n} - p^{(i-1)}_{\text{U}, n}) B \log_2 e. \label{eq:Rk_lb_pwr}	
\end{align}
By substituting ${R}_k(\mathbf{q}_n, p_{\text{U}, n})$ with ${R}_{\mathbf{p}, k}^{(i)}(p_{\text{U}, n})$, ${\Psi}^{(i)}_{\mathbf{p}, m}(\mathbf{p}_{\text{U}})$ is obtained as a tight convex upper bound of $\Psi(\mathbf{Q}^{(i)}, \bm{\beta}^{(i)}, \mathbf{p}_{\text{U}})$ and $\textbf{P}_{\text{6}}^{(i)}$ is tightened as
\begin{subequations}
	\begin{align}
		&\textbf{P}_{\text{7}}^{(i)}:  \min_{\mathbf{p}_{\text{U}}, \phi} \phi \ \ \ \ \text{s.t. } \eqref{eq:power}, \eqref{eq:radar_th2},\eqref{eq:data_availability2}, {\Psi}^{(i)}_{\mathbf{p},m}(\mathbf{p}_{\text{U}}) \leq \phi , m \in \mathcal{M}, \nonumber
	\end{align}
\end{subequations}
which can be optimally solved via IPM.

\subsection{Convergence and Complexity Analysis}
\label{sec:convergence}
Let $\eta\left(\bm{\beta}, \mathbf{Q}, \mathbf{p}_{\text{U}}\right) \! \triangleq \! \max_{\forall m} \Psi_{m}\left(\mathbf{Q}, \bm{\beta}, \mathbf{p}_{\text{U}}\right)$. By contradiction, the optimal $\phi = \eta \left(\bm{\beta},\mathbf{Q},\mathbf{p}_{\text{U}}\right)$ for $\textbf{P}_{2}$ given any feasible $\bm{\beta}, \mathbf{Q}$, and $\mathbf{p}_{\text{U}}$. Since \eqref{eq:phi2}-\eqref{eq:data_availability2} are tightened in $\textbf{P}_5^{(i)}$ and $\textbf{P}_7^{(i)}$, each update on $\mathbf{Q}$ or $\mathbf{p}_{\text{U}}$ produces a feasible solution for $\textbf{P}_2$.
Besides, as $\bm{\beta}^{(i)}$ and $\mathbf{Q}^{(i)}$ are optimal for $\textbf{P}^{(i)}_3$ and $\textbf{P}^{(i)}_5$, respectively, $\eta(\bm{\beta}^{(i-1)}, \mathbf{Q}^{(i-1)}, \mathbf{p}_{\text{U}}^{(i-1)}) \geq \eta(\bm{\beta}^{(i)}, \mathbf{Q}^{(i-1)}, \mathbf{p}_{\text{U}}^{(i-1)}) = \max_{\forall m} {\Psi}^{(i)}_{\mathbf{Q}, m}(\mathbf{Q}^{(i-1)}) \geq  \max_{\forall m} {\Psi}^{(i)}_{\mathbf{Q}, m}(\mathbf{Q}^{(i)})$. Also,
\vspace{-0.5em}
\begin{align}
	\max_{\forall m} {\Psi}^{(i)}_{\mathbf{Q}, m}(\mathbf{Q}^{(i)}) & \geq  \eta(\bm{\beta}^{(i)}, \mathbf{Q}^{(i)}, \mathbf{p}_{\text{U}}^{(i-1)}) \overset{{(a)}}{\geq} \max_{\forall m} {\Psi}^{(i)}_{\mathbf{p},m}(\mathbf{p}^{(i)}_{\text{U}}) \nonumber \\
	& \geq  \eta(\bm{\beta}^{(i)}, \mathbf{Q}^{(i)}, \mathbf{p}_{\text{U}}^{(i)}), \label{eq:eta}
\end{align}
where $(a)$ holds as $\mathbf{p}^{(i)}_{\text{U}}$ is optimal for $\textbf{P}_{\text{7}}^{(i)}$ given $\bm{\beta}^{(i)}$ and $\mathbf{Q}^{(i)}$. Thus, $\eta(\bm{\beta}^{(i-1)}, \mathbf{Q}^{(i-1)}, \mathbf{p}_{\text{U}}^{(i-1)}) \geq \eta(\bm{\beta}^{(i)}, \mathbf{Q}^{(i)}, \mathbf{p}_{\text{U}}^{(i)})$, i.e., the objective value of $\textbf{P}_{\text{2}}$ is non-increasing over iterations and convergence is thus shown. 
Upon convergence at the $i^{\star}$-th BCD iteration, a feasible solution for $\textbf{P}_1$ is constructed by setting $\beta^{\star}_{k^{\prime}, n} \triangleq \frac{\mathcal{I}_{k^{\prime}} D_{m}}{\tilde{A}_{k^{\prime}}(\mathbf{Q}^{(i^{\star})}, \bm{\beta}^{(i^{\star})}, \mathbf{p}^{(i^{\star})}_{\text{U}})} \beta^{(i^{\star})}_{k^{\prime}, n}, n\in \mathcal{N}, k^{\prime} \in \{k | \tilde{A}_{k}(\mathbf{Q}^{(i^{\star})}, \bm{\beta}^{(i^{\star})}, \mathbf{p}^{(i^{\star})}_{\text{U}}) > \mathcal{I}_k D_m, k\in \Omega_m, m\in \mathcal{M} \}$.

The block diagram of the proposed algorithm is depicted in Fig. \ref{fig:block_diagram}. 
Since the worst-case complexity of IPM is $\mathcal{O}(r^{3.5})$, where $r$ denotes the number of variables, the complexity of the proposed algorithm in each BCD iteration is $\mathcal{O}(r_a^{3.5} + r_b^{3.5} + r_c^{3.5})$ with $r_a \triangleq 3KN+M+N+K+1, r_b \triangleq 3KN+M+3N+K+1, r_c \triangleq KN+M+2N+K+1$.

\vspace{-0.5em}
\section{Simulation Results}
We simulate a UAV-assisted wireless network, where $M=2$ CMs are trained by data samples from $K=5$ IoT devices. Specifically, we train a ResNet-56 model for the CIFAR-10 dataset and a five-layer convolutional neural network for the Fashion-MNIST dataset. The two datasets are available at devices in $\Omega_1 = \left\{1, 2\right\}$ and $\Omega_2 = \left\{3, 4, 5\right\}$, respectively ($D_{1} = 24584$ bits and $D_{2} = 6276$ bits). 
MEC server and the sensing target are respectively located at $[1.7\ \text{km}, 2.9\ \text{km}]^{\mathrm{T}}$ and $[1.9\ \text{km}, 2.8\ \text{km}]^{\mathrm{T}}$, and the horizontal coordinates of the IoT devices are given by $[2.2\ \text{km}, 3.1\ \text{km}]^{\mathrm{T}}$,$[2.0\ \text{km}, 2.9\ \text{km}]^{\mathrm{T}}$, $[2.2\ \text{km}, 2.65\ \text{km}]^{\mathrm{T}}$, $[1.8\ \text{km},$ $3.1\ \text{km}]^{\mathrm{T}}$, $[1.7\ \text{km}, 2.6\ \text{km}]^{\mathrm{T}}$. 
We set $H=40~$m, $v_{\text{max}}=60~$m\slash s, $\alpha_{\text{SI}}=-110~$dB, $B=0.2~$MHz, $\lambda_0 = -50~$dB,
$N_0 = -79~$dBm, $p_{\text{UAV}}=0.04~$W, $p_k=0.01~$W, $\forall k$, $N_a = 8$, $\delta = 1~$s, $\lambda=90~$mm, $d_{p, q}=\frac{\lambda}{2}(N_a+q-p)$, $\xi \!=\! 20~\text{m}^2$, and the algorithm terminates when the difference of objective values between BCD iterations falls below $\epsilon=10^{-3}$. 
Parameters $a_1 = 25.03$, $b_1 = 0.55$, $a_2 =0.82$, $b_2 =0.22$ are acquired with $A_1=5120$ and $A_2=800$ historical data samples, respectively. The numbers of data samples available at devices are set as 1500, 2800, 800, 800, 800, respectively. We consider the following baselines for comparisons:
\begin{itemize}
	\item \textbf{Throughput Maximization (TMax)~\cite{kaitao_isac_2023}:} The uplink time allocation, UAV trajectory, and UAV transmit power are optimized with the objective of minimum throughput among devices subject to the sensing requirement.
	\item \textbf{Constant UAV Transmit Power (ConstP):} The UAV transmit power is fixed as $p_{\text{UAV}}$, while the UAV trajectory and uplink time allocation are optimized with the proposed algorithm subject to the sensing requirement.
\end{itemize}

The learning performance, including the objective value of $\textbf{P}_{\text{1}}$ and the maximum error of two CMs, is shown in Fig. \ref{fig:obj}. On one hand, \textbf{TMax} exhibits the worst classification accuracy, as it is agnostic to the different data utilities for training CMs. On the other hand, since the UAV transmit power is fixed in \textbf{ConstP}, interference arising from the sensing echo to the uplink data signal reduces the achievable rates, resulting in fewer data samples collected during the UAV's flight and thus inferior model accuracy compared with the proposed algorithm. From Fig. \ref{fig:interplay}, the tradeoff between the classification accuracy and sensing performance can be observed, both of which can be improved with a longer flight period of the UAV.

The UAV trajectory and proportion of data samples at IoT devices being collected under different schemes are shown in Fig. \ref{fig:traj40} and Fig. \ref{fig:traj100}. When $T=40~$s, the UAV cruises around to collect data samples from IoT devices with \textbf{TMax}; while under \textbf{ConstP} and the proposed algorithm, the UAV mainly collects data samples from devices in $\Omega_1$. This is because \textbf{TMax} are throughput-oriented and the other two schemes are learning-oriented. 
When $T=100~$s, under the proposed algorithm, the UAV flies closer to IoT device $1$ for improving the accuracy of CM $1$. This is because all data samples stored at IoT device $2$ can be collected with a longer flight period. Since the co-channel  interference between downlink sensing and uplink data signal may result in severe performance degradation, joint optimization for uplink time allocation, UAV trajectory, and downlink transmit power with learning-awareness is vital to dual-functional wireless networks.
\begin{figure}[t!]
	\centering
	\hspace{0.3em}
	\begin{subfigure}[b]{0.42\linewidth}
		\centering
		\includegraphics[width=\linewidth]{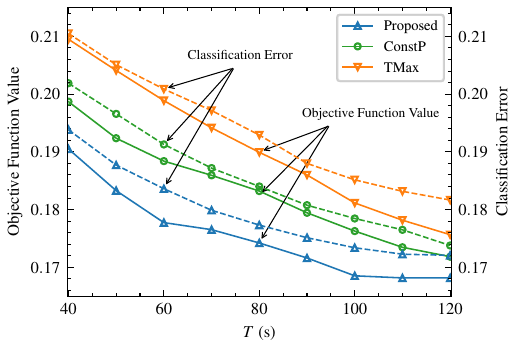}
		\subcaption{Leanring Performance}
		\label{fig:obj}
	\end{subfigure}
	\begin{subfigure}[b]{0.42\linewidth}
		\centering
		\includegraphics[width=0.9\linewidth]{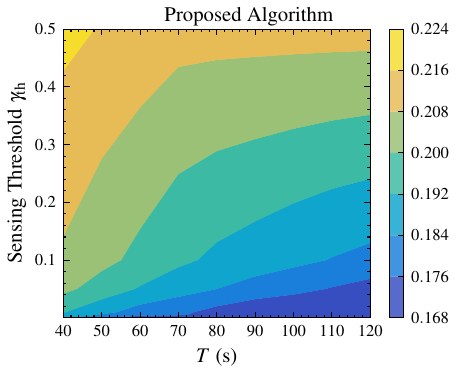}
		\subcaption{Model Error vs. $T$ and $\gamma_{\text{th}}$}
		\label{fig:interplay}
	\end{subfigure}
	\begin{subfigure}[b]{0.42\linewidth}
		\centering
		\includegraphics[width=0.85\linewidth]{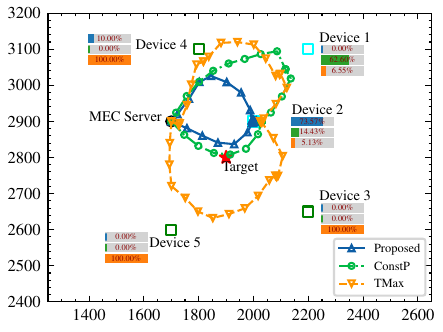}
		\subcaption{Trajectory ($T=40~$s)}
		\label{fig:traj40}
	\end{subfigure}
	\begin{subfigure}[b]{0.42\linewidth}
		\centering
		\includegraphics[width=0.85\linewidth]{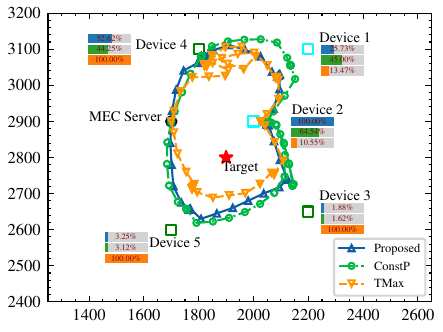}
		\subcaption{Trajectory ($T=100~$s)}
		\label{fig:traj100}
	\end{subfigure}
	\caption{.~~Learning performance and UAV trajectory of different algorithms. The progress bars in (c) and (d) represent the proportions of data samples being collected at each device.}
	\label{fig:error_compare}
	\vspace{-1.5em}
\end{figure}

\section{Conclusion}
A full-duplex UAV-enabled wireless network with sensing and edge learning functions was proposed in this letter. To maximize the accuracy of trained classification models, we developed a learning-aware optimization algorithm for uplink time allocation, UAV trajectory, and UAV transmit power. Simulation results show that the proposed algorithm outperforms the conventional data-oriented design and demonstrate the importance of UAV transmit power optimization. Extending the proposed algorithm to further enable downlink communication would be very interesting.

\bibliographystyle{IEEEtran}
\bibliography{ref}
\end{document}